\documentclass[11pt, a4paper, onecolumn]{article}

\usepackage{pstricks} 
\usepackage{graphicx} 
\usepackage[english]{babel}
\usepackage{amsmath, amssymb, amsfonts,  bm}
\usepackage{mathrsfs}
\usepackage[mathcal]{euscript}
\usepackage{enumerate}
\usepackage{algorithm}
\usepackage{algpseudocode}

\usepackage{float}
\usepackage[left=1.6cm, right=1.6cm, top=1.6cm, bottom=2cm]{geometry}

\usepackage{color,graphicx,epstopdf}
\usepackage{dsfont}
\usepackage{hyperref}
\usepackage{cite} 
\usepackage{algorithm} 
\usepackage{tikz-cd}
\usepackage{mathtools}
\usepackage{indentfirst} 
\usepackage{subcaption} 

\usepackage{amsmath}
\usepackage{setspace}
\newtheorem{theorem}{Theorem}
\newtheorem{definition}{Definition}
\newtheorem{lemma}{Lemma}

\newtheorem{remark}{Remark}

\newtheorem{assumption}{Assumption}

\DeclareFontFamily{OT1}{pzc}{}
\DeclareFontShape{OT1}{pzc}{m}{it}{<-> s * [1.200] pzcmi7t}{}
\DeclareMathAlphabet{\mathpzc}{OT1}{pzc}{m}{it}

\newcommand*\mcapinn[2]{\vcenter{\hbox{$\mathsurround=0pt
  \ifx\displaystyle#1\textstyle\else#1\fi\bigcap$}}}

\newcommand*\mcupinn[2]{\vcenter{\hbox{$\mathsurround=0pt
  \ifx\displaystyle#1\textstyle\else#1\fi\bigcup$}}}

\def\begequarr{\begin{eqnarray}}
\def\endequarr{\end{eqnarray}}
\def\begequarrs{\begin{eqnarray*}}
\def\endequarrs{\end{eqnarray*}}
\def\begequ{\begin{equation}}
\def\endequ{\end{equation}}
\def\begequs{\begin{equation*}}
\def\endequs{\end{equation*}}
\def\begite{\begin{itemize}}
\def\endite{\end{itemize}}

\def\begcen{\begin{center}}
\def\endcen{\end{center}}
\def\begrem{\begin{remark}\rm}
\def\endrem{\end{remark}}
\def\ba{\begin{aligned}}
\def\ea{\end{aligned}}

\newcommand{\mG}{\mathrm{G}}

\newcommand{\bb}{\mathbf{b}}

\newcommand{\qb}{\mathbf{q}}

\newcommand{\xb}{\mathbf{x}}
\newcommand{\yb}{\mathbf{y}}
\newcommand{\zb}{\mathbf{z}}

\newcommand{\Hb}{\mathbf{H}}
\newcommand{\Hc}{\mathcal{H}}
\newcommand{\Lb}{\mathbf{L}}
\newcommand{\Lc}{\mathcal{L}}
\newcommand{\Ib}{\mathbf{I}}

\newcommand{\Qb}{\mathbf{Q}}

\newcommand{\Jb}{\mathbf{J}}

\newcommand{\Qc}{\mathcal{Q}}
\newcommand{\Sc}{\mathcal{S}}

\newcommand{\Tb}{\mathbf{T}}
\newcommand{\Tc}{\mathcal{T}}

\newcommand{\yl}{\overline{\yb}}

\newcommand{\Hl}{\overline{\mathcal{H}}}

\newcommand{\psib}{\bm{\psi}}
\newcommand{\sbo}{\mathbf{s}^\ast}
\newcommand{\sil}{\overline{\boldsymbol{\sigma}}}
\newcommand{\ych}{\mathbf{\overline{y}}^{\perp}_{t}}
\newcommand{\ypi}{\mathbf{\overline{y}}^{\parallel}_{t}}
\newcommand{\ychd}{\mathbf{{\overline{y}}}^{\perp}_{t+1}}
\newcommand{\ypid}{\mathbf{{\overline{y}}}^{\parallel}_{t+1}}
\newcommand{\sigmab}{\boldsymbol{\sigma}}
\newcommand{\onb}{\mathbb{I}}

\title{\bf Distributed Nash Equilibrium Seeking\\ with Logarithmic Bit Rates over Digital Channels 
}

\author{%
Zihao Ren,
Chengyang Jiang,
Lei Wang,
Yang Liu, and
Kemi Ding%
\thanks{Z. Ren and L. Wang are with the College of Control Science and Engineering, Zhejiang University, China (Email: zhren2000@zju.edu.cn, lei.wangzju@zju.edu.cn). C. Jiang and K. Ding are with the Department of Automation and Intelligent Manufacturing (AIM), Southern University of Science and Technology, China (Email: 12433001@mail.sustech.edu.cn, dingkm@sustech.edu.cn). Y. Liu is with the Department of Security Research, ByteDance, China (Email: liuyang.fromthu@bytedance.com). Corresponding author: Lei Wang.}
}

   \date{}
\begin{document}
\maketitle


\begin{abstract}

This paper introduces quantization techniques to reduce the communication complexity in the distributed Nash equilibrium (NE) seeking problem, achieving an exponential reduction in bit rates over digital channels. The goal of distributed NE seeking algorithms is to coordinate agents in a network game toward equilibrium through iterative message exchanges among them via a communication network. The computational complexity of this distributed algorithm critically depends on network communication overhead in the digital channel, motivating the development of communication reduction mechanism. Regarding this, we proposed some quantizers based on sparsification and uniform quantization through a general class of ultimate-boundedness-based quantizers. Based on this, we propose a Passivity-Based NE seeking Algorithm with Time-varying scaling Error state Quantization (PBA-TEQ), and show that the linear convergence can be achieved under a sufficient condition. Moreover, when employing either the scalarization quantizer or the greedy quantizer, both belonging to  the ultimate-boundedness-based quantizers, within the PBA-TEQ framework, we establish a lower bound on communication complexity of $\log_2(\mathcal{O}(nd))$ bit rates per transmission to achieve unbiased linear convergence, with $n$ being the number of agents and $d$ being the dimension of the decision state of the network game. Numerical simulation examples are provided to validate our theoretical results.


\end{abstract}

\section{Introduction}\label{sec.Introduction}

Game theory has wide applications in areas such as generative adversarial networks, artificial intelligence, robust/resilient control problems, power allocation and even social policy-making \cite{GT,GTC,DGTW,AGTA,OLPS,NGNG,GTPA}. In these applications, network games involving multiple agents connected through a network, such as smart grids and drone swarms, have been extensively studied. For such games, agents and their communication links are denoted by nodes and edges in a network, respectively. Each agent maintains its own decision vector and a cost function, the value of which is collectively determined by the decision vectors of all agents. In this context, centralized algorithms that rely on global observations or a central agent, such as those in \cite{DAFN} and \cite{DIRA}, become ineffective. Instead, distributed Nash equilibrium (NE) seeking algorithms, which operate under partial decision information and depend on information communication through the network, have gained significant attention.


In general, the distributed NE seeking algorithms can be divided into continuous-time and discrete-time types~\cite{surv}. The work in \cite{AAFF} studies the continuous-time leader-following consensus-based flow, in which each agent's estimate of the global decision state is updated by both network consensus and the local decision vector, as the latter is updated based on the gradient of the cost function. Another continuous-time flow is the passivity-based flow studied in \cite{APAT,OTEC}. In this flow, the estimate vector is updated based on the cost function and the network consensus, as the component corresponding to this agent in the estimate vector acts as the agent's decision vector.
For discrete-time algorithms, the first focus is on the discretization of the above two flows  \cite{FDNE}. Other algorithms based on gossip communication or ADMM are proposed in \cite{DNESA,DNESU}.
Furthermore, the distributed NE seeking problems are extended to stochastic and constrained cases in \cite{DLSG,PMFC,YP-con,Lei-sto}. 
Interested readers can refer to \cite{surv} for more information on NE seeking algorithms.

For distributed computing, communication in channels over the network among agents is of great importance for the effectiveness of algorithms. 
To address this issue, extensive research has been conducted on reducing the communication complexity. The research mainly focuses on two aspects, sparsification and quantization. Sparsifiers aim to reduce the dimension of vectors for transmission. 
Specific ones, including scalarization sparsifiers and greedy sparsifiers, have been proposed in \cite{JMLR:CONT,LW-DSFN}. When sparsifiers are used, the transmitted messages still contain high-precision data types such as doubles.   
On the other hand, the study on uniform quantization always focuses on quantizing each component of the transmitted vector, respectively \cite{DACW}. 
Several methods have been proposed to reduce the quantization error based on a time-varying boundary, weighted updating, stochastic communication, adaptive quantization, and error feedback in \cite{LCOC,AR-AEQD,CROD,FCRO,CGMW}, respectively. In addition, the research on quantization has also gradually been incorporated into the distributed NE seeking problem in \cite{PAOG,UQC}. The uniform quantizers convert the transmitted message from an analog vector to a digital form with the same dimension. Therefore, the required transmission bit rate increases linearly with the dimension of the transmitted vector. Even with the sign quantizer, which  quantizes each dimension of the transmitted vector using only one bit and has been applied to distributed algorithms \cite{XY-CCFD,LCOC}, still requires $D$ bits to transmit a $D$-dimensional vector.

In summary of the previous analysis, for sparsifiers, though the dimension of the transmitted vector is reduced, the message still contains high-precision data types. In contrast, for uniform quantizers, the dimension of the transmitted vector remains unchanged, resulting in the bit rate being linearly related to the dimension of the transmission vector. When the channel bandwidth is seriously limited and the dimension of the transmitted message is large, both sparsifiers and uniform quantizers above result in prolonged communication times, which in turn reduces the efficiency of the algorithm. In order to cope with such bandwidth-limited situations, this paper aims to introduce new quantizers to further reduce the transmitted bit rate.

In this paper, we propose new quantizers, and apply them to  distributed NE seeking algorithms for an exponential reduction in the bit rate required over the digital channels. Specifically, we first propose a general class of ultimate-boundedness-based quantizers, that inspires scalarization and greedy quantizers. Next, we present the Passivity-Based NE seeking Algorithm with Time-varying scaling Error state Quantization (PBA-TEQ). By incorporating the ultimate-boundedness-based quantizers into the PBA-TEQ framework, we derive a sufficient condition that guarantees linear convergence. Furthermore, we establish a conservative lower bound of $\log_2(\mathcal{O}(nd))$ on the bit rate required per transmission when using either scalarization or greedy quantizers, with $n$ being the number of agents and $d$ being the dimension of the decision state of the network game. We also compare the sparsifiers \cite{JMLR:CONT,LW-DSFN} and the sign quantizer \cite{XY-CCFD, LCOC} from the literature with the proposed quantizers.
{
Note that due to the presence of quantization, the convergence rate of our algorithm is slower compared with those without quantization in the literature \cite{APAT}. However, as shown by numerical simulations, the  proposed algorithm shows a significant reduction in the total bit rate required to achieve the same level of accuracy.}

This paper makes two main contributions as follows. 
\begin{itemize}

  \item A general class for ultimate-boundedness-based quantizers---that inspires two new specific quantizers, the scalarization quantizer and the greedy quantizer---is developed to reduce communication overhead over digital channels in the distributed NE seeking problem.  
    
    \item Sufficient conditions are established to achieve linear convergence when incorporating the ultimate-boundedness-based quantizers into the distributed NE seeking algorithm. Furthermore, for specific quantizers, a lower bound is derived of $\log_2(\mathcal{O}(nd))$ on the bit rate required per transmission. In comparison to the sign quantizers \cite{XY-CCFD, LCOC} requiring $nd$ bits per transmission and the unquantized sparsifiers \cite{JMLR:CONT,LW-DSFN} requiring high-precision data transmission, our results offer significant potential of reducing communication complexity in distributed algorithms.

\end{itemize}

\emph{\bf Notation.}
In this paper, $\|\cdot \|$
denotes the Euclidean norm. The notation $\mathbf{1}_{n}(\mathbf{0}_{n})$, $\mathbf{1}_{n\times d}(\mathbf{0}_{n\times d})$, $\mathbf{I}_n$ and $\{{\bf e}_1,\dots,{\bf e}_d\}$ denote the one (zero) column, the one (zero) matrix, identity matrix and base vectors in $\mathbb{R}^d$, respectively. The expression $\mathrm{blkdiag}(\xb_1,\dots,\xb_n)$ is a diagonal matrix with the $i$-th diagonal matrix being $\xb_i$. The symbol $\otimes$ denotes the Kronecker product. For a differential function, $\nabla(\cdot)$ denotes its gradient. For column vectors $\mathbf{a}$ and $\mathbf{b}$, $[\mathbf{a};\mathbf{b}]$ means $[\mathbf{a}^{\top},\mathbf{b}^{\top}]^{\top}$. The notation $\mathcal{O}(\cdot)$ means the magnitude notation.

\section{Problem Formulation}
\label{sec.pro}
\subsection{Game Theory}

In this paper, we consider a network game involving 
$n$ agents indexed by $\mathbb{V}=\{1,2,\dots, n\}$. 
In this network game, each agent 
$i\in\mathbb{V}$ selects its action, which can be represented by a 
vector 
$\xb_i\in\mathbb{R}^d$. The vector 
$\xb^0=[\xb^{\top}_1,\dots,\xb^{\top}_n]^{\top}\in\mathbb{R}^{nd}$ represents the actions of all agents. 
Each agent 
$i$ also holds a $C^2$ cost function 
$f_i:\mathbb{R}^{nd}\rightarrow \mathbb{R}$,
whose value is determined by the actions $\xb^0$ of all agents. 
The goal of each agent is to choose its best action to minimize its cost function as much as possible. 
Without loss of generality, we introduce the following assumption regarding the cost functions.

We denote this game by $\mathcal{G}(\mathbb{V},\{f_i|i\in\mathbb{V}\})$. The optimal solution in the context of this non-cooperative game is known as the NE \cite{JNASH}, which is defined as follows.

\begin{definition}\label{def:NE}
    An NE is an action profile for all agents, where no agent can reduce its cost function by unilaterally changing its action, that is,  
$\xb^\ast=[\xb^{\ast T}_1,\dots,\xb^{\ast T}_n]^{\top}\in\mathbb{R}^{nd}$
  is an NE  of the game $\mathcal{G}$ if 
  \[
  \ba
    f_i({\xb^\ast_{i}},{\xb^\ast_{-i}})\leq f_i({\xb_{i}},{\xb^\ast_{-i}}) \quad\forall  \xb_{i}\in\mathbb{R}^d,i\in \mathbb{V},
    \ea
  \]
  where $\xb^\ast_{-i}:=[{\xb^\ast_{1}};\dots,{\xb^\ast_{i-1}};{\xb^\ast_{i+1}};\dots,{\xb^\ast_{n}}]$.\hfill  
\end{definition}

Next, 
we define the gradient function
$\Hb(\xb^0):=[\frac{\partial f_1}{\partial \xb_1}(\xb^0);\dots,;\frac{\partial f_n}{\partial \xb_n}(\xb^0)]\in\mathbb{R}^{nd}$ and introduce the following assumption to ensure the existence and uniqueness of NE for game $\mathcal{G}$.

\begin{assumption}\label{ass-H1}
    $\Hb$ is strongly monotone, i.e., 
    \begin{equation}
        \label{eq:ass-H1}
        (\xb-\xb')^{\top}(\Hb(\xb)-\Hb(\xb'))\geq \mu \|\xb-\xb'\|^2,
    \end{equation}
    for any $\xb,\xb'\in\mathbb{R}^{nd}$ and some $\mu>0$ and Lipschitz continuous, i.e., $\|\Hb(\xb)-\Hb(\xb')\|\leq L_\Hb \|\xb-\xb'\|$ for any $\xb,\xb'\in\mathbb{R}^{nd}$ and some $L_\Hb>0$.\hfill  
\end{assumption}
  
According to \cite{RACM}, there exists a unique NE $\sbo\in\mathbb{R}^{nd}$ for the network game $\mathcal{G}$ if Assumption \ref{ass-H1} holds.

 In addition, we define the extended gradient as \( \Hc(\yb) := \left[ \frac{\partial f_1}{\partial \xb^1_1}(\xb^1); \dots; \frac{\partial f_n}{\partial \xb^n_n}(\xb^n) \right] \in \mathbb{R}^{nd} \) for \( \yb \in \mathbb{R}^{nnd} \). Similarly to \cite{APAT, OTEC}, we introduce the following assumption for subsequent analysis.

\begin{assumption}\label{ass-H2}
    $\Hc$ is Lipschitz continuous, i.e., 
    \begin{equation}
        \label{eq:ass-H2}
        \|\Hc(\yb)-\Hc(\yb')\|\leq L_\Hc \|\yb-\yb'\|,
    \end{equation}
 for any $\yb,\yb'\in\mathbb{R}^{nnd}$ and some $L_\Hc>0$.\hfill  
\end{assumption}

\subsection{Distributed Nash Equilibrium Seeking}
\label{sec.PBNSF}

In this paper, agents of this game will interact with each other through a communication network described by a graph $\mathrm{G}=(\mathbb{V},\mathbb{E})$, where 
$\mathbb{E}$ denotes the set of edges. Let $[a_{ij}]\in \mathbb{R}^{n\times n}$ denote the weight matrix complying with graph $\mathrm G$, i.e., $a_{ij}>0$ if $\{j,i\}\in\mathbb{E}$ and $a_{ij}=0$ if $\{j,i\}\notin\mathbb{E}$. Moreover, we use $\mathbf{L}$ to represent the Laplacian matrix of graph $\mathrm G$, which satisfy $[\Lb]_{ij}=-a_{ij}$ for all   $i\neq j$, and $[\Lb]_{ii}=\sum_{j=1}^n a_{ij}$ for all $i\in\mathbb{V}$. The set of neighboring agent for agent $i$ is denoted by $\mathbb{N}_i$, and $j\in\mathbb{N}_i$ if and only if $[\Lb]_{ij}\neq0$ for all $i,j\in\mathbb{V}$. 
We introduce the following assumption for the graph $\mathrm{G}$.
\begin{assumption}
    \label{ass:graph}
    The graph $\mathrm{G}$ is undirected, connected and time-invariant.
\end{assumption}

Note that if Assumption \ref{ass:graph} holds, the Laplacian matrix $\mathbf{L}$ is symmetric and positive semi-definite, with eigenvalues $\lambda_i$, $i\in\mathbb{V}$ in an ascending order satisfying $0=\lambda_1<\lambda_2\leq\dots\leq \lambda_n$ and $\mathbf{1}_n^{\top} \Lb=0$ by \cite{magnusbook}. 

In a distributed NE seeking algorithm, each agent $i$ holds its cost function $f_i$ and an estimate of the actions of all agents, denoted by $\xb^i := [\hat{\xb}_{1}^i; \dots; \hat{\xb}_n^i] \in \mathbb{R}^{nd}$, where $\hat{\xb}_i^i = \xb_i$ represents its actual action. The agents share their estimates over the graph and update them based on received messages and the gradients of the cost functions. The objective of distributed NE seeking algorithms is to drive the agents' action states toward the NE.


\subsection{Problem of Interest}

In the implementation of distributed NE-seeking algorithms, communication complexity is a critical factor in determining  efficiency, particularly in scenarios where the bandwidth of digital channels is severely limited. To address this issue, sparsifiers have been introduced to reduce the dimension of transmitted messages. However, they still require the transmission of high-precision data types. Uniform quantizers, including sign quantizers, on the other hand, necessitate at least $nd$ bits to communicate an $nd$-dimensional vector. Motivated by these limitations, this paper investigates a key question in distributed NE seeking over bandwidth-limited communication networks: How can one design quantizers that substantially reduce the required transmitted bit rate, thereby reducing communication complexity beyond what is achievable with sparsification or uniform quantization?

\section{Communication Quantization}
\label{qm}

In this section, we introduce a general class of quantizers, and thus proposing 
two new quantizers.

First, drawing inspiration from the ultimate boundedness theory \cite[Theorem 4.18]{Khalil(2002)}, we introduce the following general class for ultimate-boundedness-based quantizers.

\begin{definition}
\label{def:com}
The mapping $\Qb:\mathbb{R}^{D}\times \mathbb{N}_+ \to \mathbb{R}^D$ is said to be an $(u, \delta)$-{\bf ultimate-boundedness-based quantizer} if
for all $t \in \mathbb{N}_+$, any $a>0$ and $\xb,\yb\in\mathbb{R}^{D}$ satisfying $\|\xb\|\leq au$, there exists a Lyapunov function $V_e:\mathbb{R}^{D}\times \mathbb{N}_+ \to \mathbb{R}^D$ and some $\kappa_0>0$ such that
\begin{equation}
\label{eq:def2}
    \ba
    &\qquad\quad c_1\|\xb\|^2\leq V_e(\xb,t)\leq c_2\|\xb\|^2,
    \\
    &V_e\left(\xb+\kappa_0a\Qb\left(\frac{\xb}{a},t\right),t+1\right)-V_e(\xb,t)\\&\qquad\qquad\qquad\quad\ \leq -c_3\|\xb\|^2-\delta^2\|a\|^2,\\
    &V_e(\xb+\yb,t)-V_e(\xb,t)\leq c_4(\|\yb\|^2+2\|\xb\|\|\yb\|),
    \ea
\end{equation}
for some $c_1,c_2,c_3,c_4>0$.

\end{definition}

According to Definition \ref{def:com}, we know the system 
$
\xb(t+1) = \xb(t) - \kappa_0a\Qb\left(\frac{\xb}{a},t\right), t)
$
is ultimately bounded within the \(\frac{a\delta}{\sqrt{c_3}}\)-neighborhood of the zero equilibrium with attraction region \( \|\xb\|\leq au \). Together, parameters $(u,\delta)$ characterize the quantization effect of \(\Qb\) on the algorithm. Specifically, if \( u \) is smaller and \( \delta \) is larger, the error introduced by quantization tends to increase. 




Before proposing quantizers based on this framework, we introduce an intuitive quantization function. The function $q_m: \mathbb{R} \to \{0,1\}^m$, for some $m \in \mathbb{N}_+$, is an $m$-bit quantization function that represents an analog value as an $m$-bit binary number. Specifically, let $l := 2^{1-m}$ be the quantization gap. Then, $q_m$ is defined as
 \[
      \ba q_m(x)=\left\{
        \ba
             1-\frac{l}{2},\ &x>1, \\
             I,\qquad\ \ &x\in (I-\frac{l}{2},I+\frac{l}{2}],\\ & 
  I=\pm \frac{1}{2}l,\pm \frac{3}{2}l,\dots,\pm (1-\frac{1}{2}l),\\
             -1+\frac{l}{2},\ &x\leq -1.
        \ea
        \right.
        \ea
    \]
Moreover, it can be noted that
\begin{equation}
\label{eq:qm}
\ba
    |q_m(x)-x|\leq \frac{l}{2},\quad {\rm if}\ |x|\leq 1.
    \ea
\end{equation}

Using the function $q_m$ and inspired by Definition \ref{def:com}, we now propose the following two quantizers.

\begin{definition}
\label{def:sq}
    {\bf The $m$-bit scalarization quantizer} $\Qb_{\mathrm{sc}}:\mathbb{R}^D\rightarrow \{0,1\}^m$ with some $m\in\mathbb{N}_+$ satisfies $\Qb_{\mathrm{sc}}(\xb_e)=\psib(t) q_m(\psib(t)^{\top}\xb_e)$ at time $t$, where the unit vector $\psib:\mathbb{N}_+\rightarrow \mathbb{R}^D$ is uniformly bounded
and persistently excited, i.e., 
        $
    \alpha_2 \Ib_d \geq  \sum_{s=t}^{t+kD} \psib(s)\psib^{\top} (s) ds \geq \alpha_1 \Ib_d,
    $
    for all $t\geq 0$ and some constants $\alpha_1,\alpha_2,k>0$. 
    
\end{definition}  

\begin{definition}
\label{def:gq}
    {\bf The $m$-bit greedy quantizer} $\Qb_{\mathrm{gr}}:\mathbb{R}^D\rightarrow \{0,1\}^{m}$ with some $m\in\mathbb{N}_+$ satisfies $\Qb_{\mathrm{gr}}=q_m([\xb_e]_{s}){\bf e}_{s}$, where $s$ is the index of the highest coordinates in the absolute value of $\xb_e$. 

    
\end{definition}

\begin{remark}

We analyze the quantization process of $ \Qb_{\mathrm{sc}} $ and $ \Qb_{\mathrm{gr}} $. For $ \Qb_{\mathrm{sc}} $, the quantized and transmitted value is a scalar $ q_m(\psib(t)^{\top} \xb_e) $, and the receiver multiplies $ \psib(t) $ to recover $ \Qb_{\mathrm{sc}} $. For $ \Qb_{\mathrm{gr}} $, the component $ [\xb_e]_s $ is quantized and transmitted, while the other components of $ \xb_e $ are omitted (i.e., set to zero). The versions of $ \Qb_{\mathrm{sc}} $ and $ \Qb_{\mathrm{gr}} $ without quantization are discussed in \cite{LW-DSFN} and \cite{JMLR:CONT}, respectively.
\end{remark}


In addition, we introduce the sign quantizer from literature \cite{XY-CCFD, LCOC} as follows.

\begin{definition}
\label{def:signq}
{\bf The sign quantizer} $\Qb_{\mathrm{si}}: \mathbb{R}^D \to \{0, 1\}^D$ satisfies $[\Qb_{\mathrm{si}}(\xb_e)]_i = q_1([\xb_e]_i)$ for each dimension $i=1,2,...,D$.
    
\end{definition}  

\begin{remark}
The quantizer 
$\Qb_{\rm si}$
  converts the transmitted vector into a vector with each dimension being $-\frac{1}{2}$ or $\frac{1}{2}$, and requires 
$D$ bits to transmit a 
$D$-dimensional vector. When applied to quantize network communication, it has been shown that a distributed algorithm with unbiased linear convergence maintains this convergence property. Notably, though several quantizers have been established in the field of communication quantization \cite{AR-AEQD,CROD,FCRO,CGMW,XY-CCFD,LCOC,PAOG,UQC}, 
$\Qb_{\rm si}$ stands out as the quantizer that minimizes the bit rate for transmission, while still ensuring unbiased linear convergence of the algorithm.
\end{remark}

We note that the  proposed two quantizers, $\Qb_{\mathrm{\mathrm{sc}}}$ and $\Qb_{\mathrm{gr}}$, together with the sign quantizer $\Qb_{\mathrm{si}}$, share common properties, as outlined in the following result.





\begin{lemma}
\label{thm:com}
$\Qb_{\mathrm{sc}}$, $\Qb_{\mathrm{gr}}$ and $\Qb_{\mathrm{si}}$ belong to the class of ultimate-boundedness-based quantizers. 
\end{lemma}

The proof of Lemma \ref{thm:com} can be seen in Appendix \ref{app:pro}.






\color{black}
\section{Main Results}
\label{sec.mai}

In this section, we propose an algorithm to address the distributed NE seeking problem with quantized communication,   based on the passivity-based NE seeking algorithm~\cite{APAT}.
{
In this algorithm, the distributed Nash equilibrium seeking problem is reformulated as a consensus problem of the agents' estimates, and each agent updates its decision based on these estimates, referring to as passivity.} 
To reduce the quantization error, inspired by \cite{LCOC}, we quantize the error state and introduce a decaying coefficient \( \gamma_t \) to scale the transmitted value before and after quantization, as illustrated in Fig.\ref{fig:illu0}. This results in the Passivity-Based NE seeking Algorithm with Time-varying scaling Error state Quantization (PBA-TEQ).

\begin{figure}
    \centering\includegraphics[width=1\linewidth]{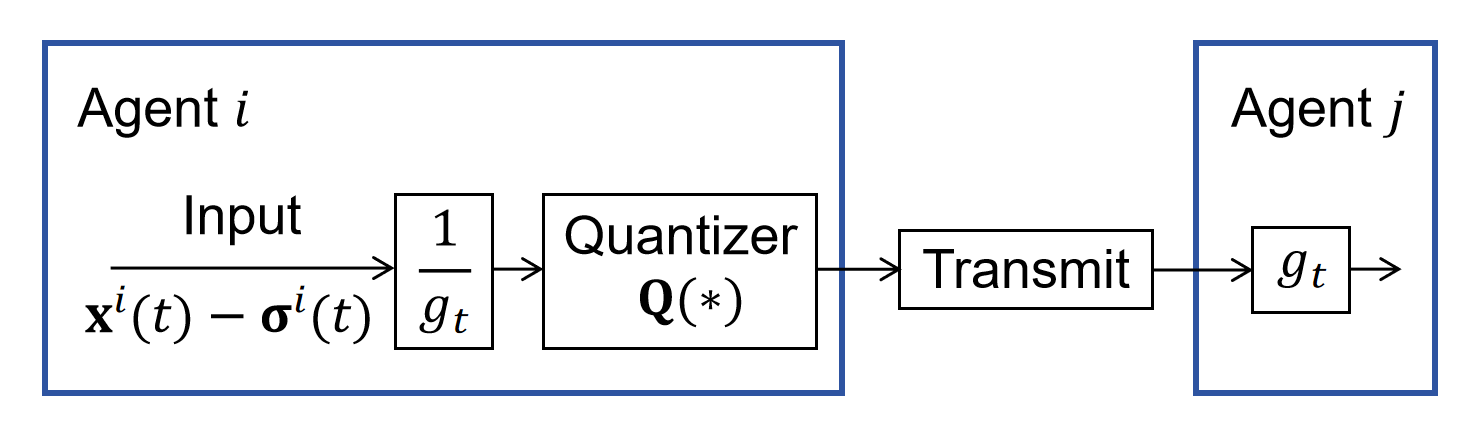}
    \caption{The illustration of time-varying scaling quantization.}
    \label{fig:illu0}
\end{figure}

\begin{algorithm}

\caption{The Passivity-Based NE seeking Algorithm with Time-varying scaling Error state Quantization (PBA-TEQ)}
\label{alg}
\begin{algorithmic}
\State {\bf Initialization}: $\kappa_0,\alpha,\eta,g_0,\gamma>0$, $ \sigmab^i(0)=\zb^i(0)=\xb^i(0)=\mathbf{0}_d, i\in\mathbb{V}$.

\For{$t \in \mathbb{N}_+$, each agent $i$}
\State Calculate $g_t=g_0{\gamma^t}$ 
     \State Transmit $\qb^i(t)=\Qb\left(\frac{\xb^i(t)-\sigmab^i(t)}{g_t},t\right)$ to neighbors
    \State Update:
    \begin{equation}
    \label{eq:PBA-TEQ}
        \ba{\sigmab}^i(t+1)=&{\sigmab}^i(t)+\kappa_0g_t\qb^i(t),\\
{\zb}^i({t+1})=&\zb^i(t)\\&+\kappa_0g_t\left[\qb^i(t)+\sum_{j\in\mathbb{N}_i}a_{ij}\left(\qb^j(t)-\qb^i(t)\right)\right],\\
{\xb}^i({t+1})=&\xb^i(t)\\&- \left[\alpha(\sigmab^i(t)-\zb^i(t))+\eta\Tb_i \nabla_i f_i(\xb^i(t))\right].
        \ea
    \end{equation}
\EndFor
\end{algorithmic}
\end{algorithm}

In Algorithm \ref{alg}, the vector \( \yb := [\xb^{1}; \dots; \xb^{n}] \in \mathbb{R}^{nnd} \) represents the estimate vector of all agents, the parameters \( g_0, \gamma,\alpha, \eta > 0 \) are to be determined, the matrix \( \Tb_i := [\mathbf{0}_{d \times (i-1)d};\mathbf{I}_d; \mathbf{0}_{d \times (n-i)d}] \in \mathbb{R}^{nd \times d} \) is a selection matrix, and \( \nabla_i f_i(\xb^i) := \frac{\partial f_i}{\partial \xb^i_i}(\xb^i) \in \mathbb{R}^d \). It should be noted that only the quantized values \( \Qb\left( \frac{\xb^i(t) - \sigmab^i(t)}{g_t}, t \right) \) are transmitted over the network.

Next, we integrate the quantizers into Algorithm \ref{alg} and perform a convergence analysis. Moreover, we establish the required bit rate per transmission for different quantizers.

\subsection{Convergence Analysis for the Ultimate-Boundedness-Based Quantizers}
\label{sec.conana}

First, for the general ultimate-boundedness-based quantizers, which encompass quantizers $\Qb_{\mathrm{sc}},\Qb_{\mathrm{gr}},\Qb_{\mathrm{si}}$, we have the following theorem.

\begin{theorem}
\label{thm-TSC}    
Consider the game $\mathcal{G}$ over the communication graph $\mG$ and assume that Assumptions \ref{ass-H1}, \ref{ass-H2} and \ref{ass:graph} hold. In addition, suppose $\Qb$ is $(u,\delta)$-ultimate-boundedness-based and satisfies
\begin{equation}
    \label{eq:uanddelta_O_INI}
 \frac{u}{\delta}\geq\sqrt{\frac{1+\epsilon}{\beta}}.
\end{equation}
for any $\epsilon>0$ and $\beta = \mathrm{min}\left\{\frac{\lambda_2\alpha}{6},\frac{\eta\mu}{4n},\frac{c_3}{4c_2 }\right\} \in (0, 1)$. Under these conditions, there exist constants $\alpha$, $\eta$, $g_0$ and $\gamma=\sqrt{\frac{1-\beta}{1-\beta/(1+\epsilon)}}\in(0,1)$ such that the global action of each agent $\xb^i(t)$ generated by Algorithm \ref{alg}, converges linearly to the NE of the game $\mathcal{G}$, i.e. ,
\[
\|\xb^i(t) - \xb^\ast\|^2 = \mathcal{O}(\gamma^{2t}).
\]
(see the parameters $c_3, c_2, \alpha, \eta, g_0$ in Appendix \ref{app:thm2}.)

\end{theorem}

The proof of Theorem \ref{thm-TSC} can be seen in Appendix \ref{app:thm2}.
Theorem \ref{thm-TSC} demonstrates the convergence performance when the general ultimate-boundedness-based quantizers are incorporated into Algorithm \ref{alg}. When 
$\frac{u}{\delta}$
  of the quantizer is sufficiently large, the algorithm can converge to the NE of the game at a linear convergence rate.

\subsection{Provable Bit Rates for Sign Quantizer}
For the sign quantizer \( \Qb_{\mathrm{si}} \) in Definition \ref{def:signq}, we have the following theorem.

\begin{theorem}
\label{lem-Q3}    
Consider the game \( \mathcal{G} \) over the communication graph \( \mG \) and assume Assumptions \ref{ass-H1}, \ref{ass-H2} and \ref{ass:graph} hold. Then, if \( \Qb = \Qb_{\mathrm{si}} \), there exist some constants  $\alpha$, $\eta$, $g_0$, $\gamma>0$ such that the global action of each agent $\xb^i(t)$ generated by Algorithm \ref{alg}, converges linearly to the NE of the game $\mathcal{G}$.
\end{theorem}

The proof of Theorem \ref{lem-Q3} can be seen in Appendix \ref{appc}. With Definition \ref{def:signq}, we know Theorem \ref{lem-Q3} indicates that the bit rate required in each transmission for $\Qb_{si}$ is $nd$.

\subsection{Provable Bit Rates: Exponential Reduction}
For the \( m \)-bit scalarization quantizer \( \Qb_{\mathrm{sc}} \) in Definition \ref{def:sq} and the \( m \)-bit greedy quantizer \( \Qb_{\mathrm{gr}} \) in Definition \ref{def:gq}, we have the following theorem.

\begin{theorem}
\label{thm-C1C2}    
Consider the game \( \mathcal{G} \) over the communication graph \( \mG \) and assume Assumptions \ref{ass-H1}, \ref{ass-H2} and \ref{ass:graph} hold. Then, if \( \Qb\) is the $m$-bit quantizer \(\Qb_{\mathrm{sc}}\) or \(  \Qb_{\mathrm{gr}} \) with \( m \geq \underline{m} = \log_2(\mathcal{O}(nd)) \), there exists some constant  $\alpha$, $\eta$, $g_0$, $\gamma>0$ such that the global action of each agent $\xb^i(t)$ generated by Algorithm \ref{alg}, converges linearly to the NE of the game $\mathcal{G}$.
\hfill  
\end{theorem}

The proof of Theorem \ref{thm-C1C2} can be seen in Appendix \ref{pr:Cor1}.

According to Theorems \ref{lem-Q3} and \ref{thm-C1C2}, we can establish the conservative bit rate required for different quantizers. Specifically, the linear bit rate for \( \Qb_{\mathrm{si}} \) and the logarithmic bit rate for \( \Qb_{\mathrm{sc}}, \Qb_{\mathrm{gr}} \), with respect to \( nd \), the dimension of the transmitted vector, are obtained. To illustrate, the required bit rate for each quantizer is shown in Table \ref{tab:illu} and Fig.\ref{fig:illu}.

\begin{table}[http]
    \centering
     \caption{The bit rates for different quantizers.}
    \begin{tabular}{c||c|c} Quantizer&$\Qb_{\mathrm{si}}$&$\Qb_{\mathrm{sc}},\Qb_{\mathrm{gr}}$\\ \hline Bit Rate& $nd$ & $\log_2(\mathcal{O}(nd))$
   \end{tabular}  
    \label{tab:illu}
\end{table}

\begin{figure}
    \centering\includegraphics[width=0.8\linewidth]{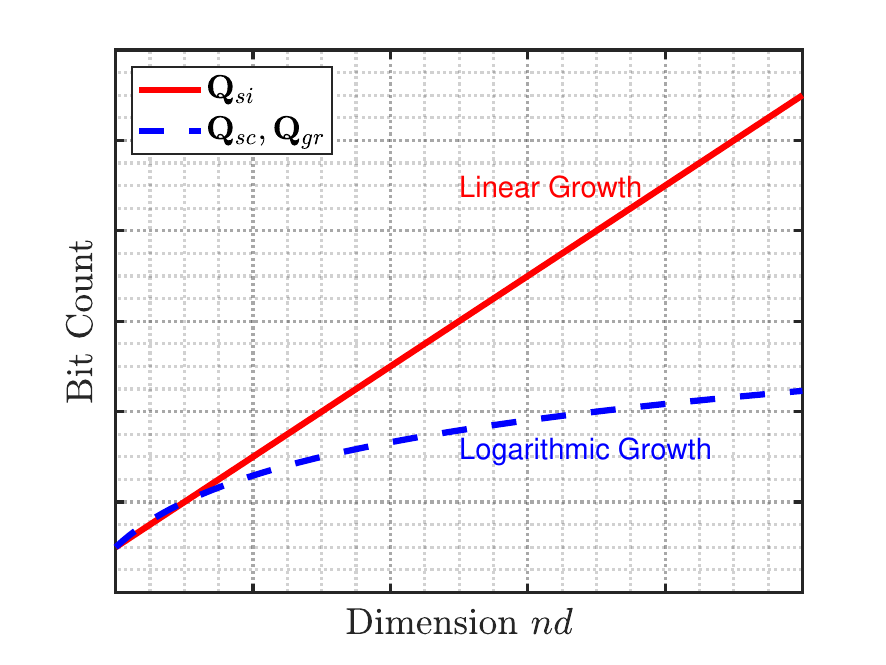}
    \caption{The bit rates for different quantizers with respect to dimension $nd$.}
    \label{fig:illu}
\end{figure}

Compared to the linear bit rate required for \( \Qb_{\mathrm{si}} \) and similar results in the existing literature \cite{XY-CCFD, LCOC}, the quantizers \( \Qb_{\mathrm{sc}} \) and \( \Qb_{\mathrm{gr}} \) reduce the bit rate in each transmission over the channels \textit{exponentially}, i.e., from linear to logarithmic, with linear convergence guaranteed. Meanwhile, compared with their unquantized sparsifier version, $\Qb_{\rm si}$ and $\Qb_{\rm gr}$ quantize high-precision data using $m$ bits, whose lower bound is given in theory. Therefore, compared with uniform quantizers and sparsifiers,  the significant reduction in bit rates established in Theorem \ref{thm-C1C2} plays a crucial role in alleviating the communication complexity in distributed NE seeking algorithms.


\section{Numerical Simulations}

In this section, we evaluate the performance of Algorithm \ref{alg} on three types of synthetic networks generated using the Erd\H{o}s–R\'{e}nyi (ER), Watts–Strogatz (WS), and Barab\'{a}si-Albert (BA) models. The networks considered for evaluation consist of $n=30$ nodes and the dimension of each decision state is $d=5$. Each agent holds a local function $f_i(\xb)=\frac{1}{2}\|\Jb^i \xb-\bb_{i}\|^2 $ with some randomly generated $\mathbf{J}^{i}\in\mathbb{R}^{nd\times nd}$ and $\bb_{i}\in\mathbb{R}^{nd}$. We select an example of $\mathbf{J}^{i}$ such that $
\Hb(\xb)=[\Jb^1_1,\dots,\Jb^n_n]\xb
$ satisfies Assumption \ref{ass-H2} with $\Jb^i_i$ being the $i$-th row of matrix $\Jb^i$. Hence, the game $\mathcal{G}$
 has a unique NE  
$\sbo\in\mathbb{R}^{nd}$. 

\textbf{Network Setup}.  In the ER graph, each link is independently established with a probability of $p_{er}=0.1$. In the WS graph, the average degree is set to $k_{ws}=5$, and the rewiring probability is $p_{ws}=0.2$. In the BA graph, a new node is added at each time step and connects to $m_{ba}=2$ existing nodes via preferential attachment. The randomly generated network structures from the ER, WS, and BA models are shown in Fig. \ref{fig-structure}.

\begin{figure}[http]
    \centering
    \includegraphics[width=5cm]{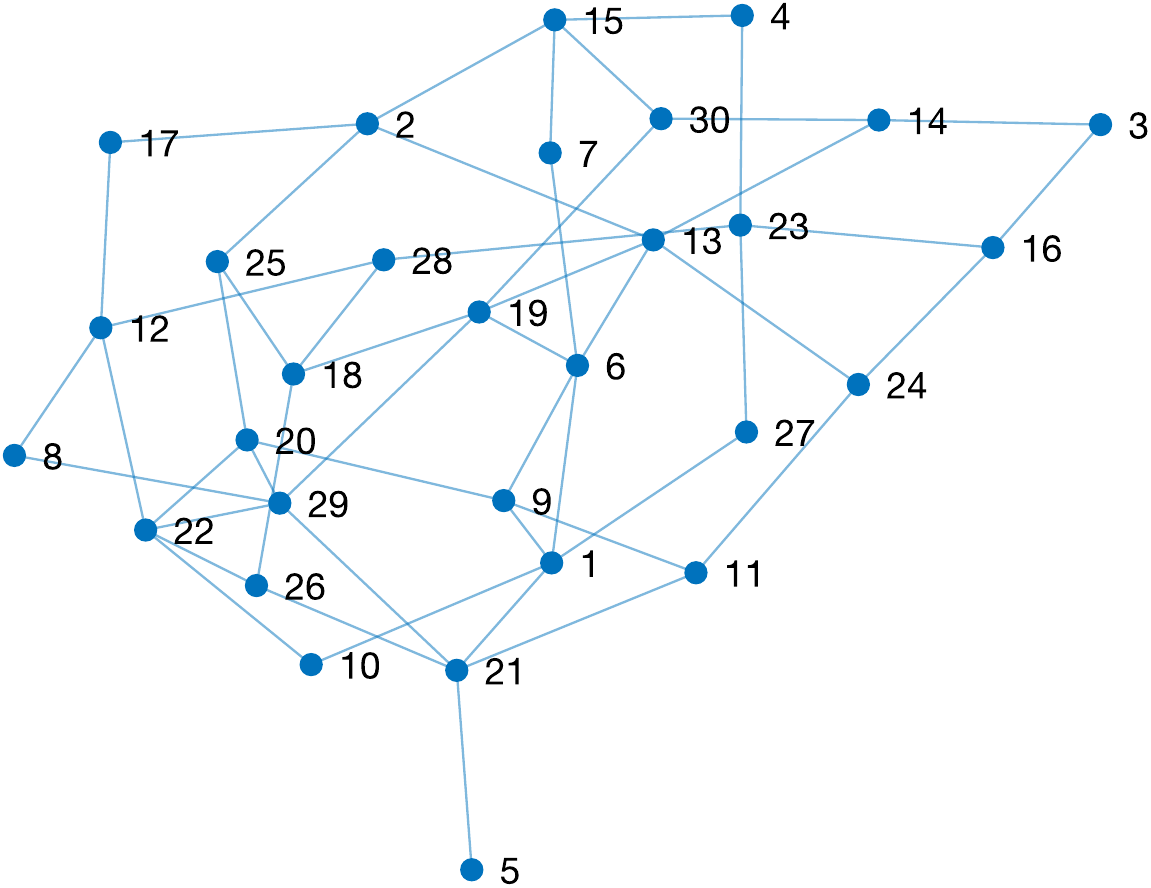}
    \begin{center}
    (a) Erd\H{o}s–R\'{e}nyi
    \end{center}
    \includegraphics[width=5cm]{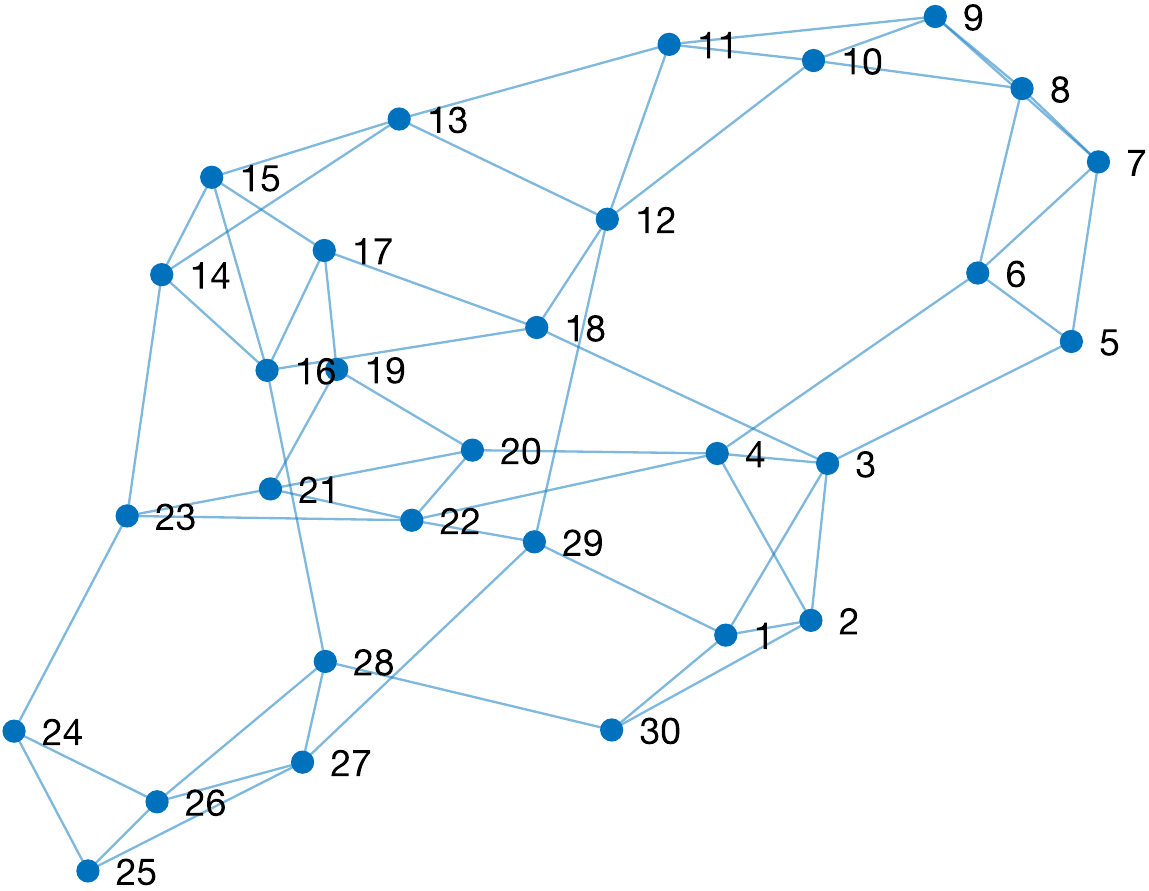}
    \begin{center}
    (b) Watts–Strogatz
    \end{center}
    \includegraphics[width=8cm]{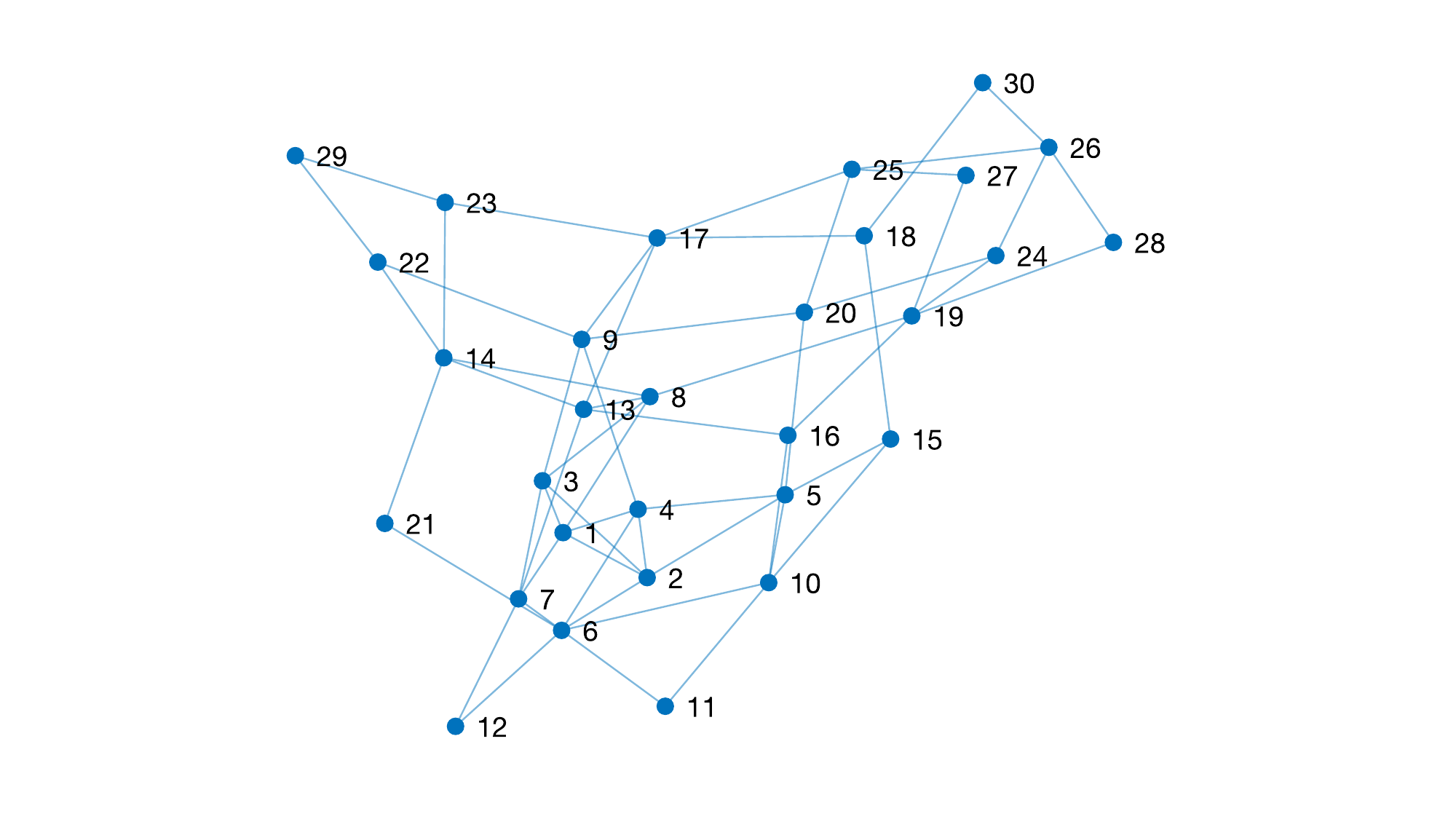}
    \begin{center}
    (c) Barab\'{a}si-Albert
    \end{center}
    \caption{Schematic illustration of network structures randomly generated using the ER, WS, and BA models.}
    \label{fig-structure}
\end{figure}

\textbf{Data Generation}. The parameters and objective functions of the agents over networks are the same as ring graph, and it ensures that the network game $\mathcal{G}$ has a unique NE.

\textbf{Comparative Analysis of Different Quantizers}. 
In addition to the quantizers $\Qb_{\rm sc},\Qb_{\rm gr}$ and $\Qb_{\rm si}$ mentioned before, we first introduce two more quantizers, the stochastic quantizer $\Qb_{\rm sto}$ \cite{QSTO} and the integer quantizer $\Qb_{\rm int}$ \cite{XY-CCFD} from the existing literature for comparison.
\[
\ba
&\Qb_{\rm sto}(\xb)=\frac{\|\xb\|_{\infty}}{2^{3}}{\rm sign}(\xb)*\lfloor\frac{2^{3}|\xb|}{\|\xb\|_{\infty}}+\overline{\omega}\rfloor,\\
&\Qb_{\rm int}(\xb)=\lfloor \xb+\frac{\mathbf{1}}{2}\rfloor,
\ea
\]
where \( * \) denotes the Hadamard product, \(\overline{\omega}\) is a stochastic vector with each element following a uniform distribution on \([0, 1]\), and \(\lfloor \cdot \rfloor\) represents the floor function.
The convergence curves of the NE seeking by Algorithm \ref{alg} with different quantizers on the network are shown in Fig. \ref{fig-net-quantizer}, where $\hat\Qb_{\mathrm{sc}}$ is an example of $\Qb_{\mathrm{sc}}$ by letting $\psib(t)= \mathbf{e}_i $ with $i=1+(t\   \mathrm{mod}\ D)$. 

Furthermore, an analysis of the convergence rate reveals an inherent trade-off between per-iteration progress and overall communication efficiency. On one hand, the sign quantizer $\Qb_{\mathrm{si}}$ transmits the full-dimensional state with high fidelity, which accelerates convergence per iteration. On the other hand, the aggressive dimension reduction employed by quantizers $\Qb_{\mathrm{sc}}$ and $\Qb_{\mathrm{gr}}$ cuts down the bit rate transmitted at the expense of reducing the amount of information exchanged in each round. Notably, quantizer $\Qb_{\mathrm{sc}}$ employs a scalarization technique that projects the high-dimensional state onto a single scalar value before quantization. This aggressive dimensionality reduction significantly limits the information conveyed in each communication round, thereby leading to smaller, less informative updates and ultimately slower convergence.

\begin{table}[http]
    \centering
       \caption{The parameters of quantizers and the relevant data of Algorithm \ref{alg} without quantizers and with quantizers $\hat{\Qb}_{\mathrm{sc}}$, $\hat{\Qb}_{\mathrm{gr}}$, $\Qb_{\mathrm{si}}$ over different network structures.}
       
    \begin{center}
    (a) Erd\H{o}s–R\'{e}nyi
    \end{center}
    \begin{tabular}{c||c|c|c|c|c|c}
    Quantizers & No $\Qb$   & $\hat{\Qb}_{\mathrm{sc}}$ & $\hat{\Qb}_{\mathrm{gr}}$& $\Qb_{\mathrm{si}}$&$\Qb_{\rm int}$& $\Qb_{\rm sto}$  \\ \hline
    Values / Iter. &150 &1 &4 &150 & 150&150 \\ \hline 
    Bits / Value &64 &2 &1 &1 &16 &3\\ \hline 
    Iterations/$10^4$ &0.70 &2.65 &1.21 &1.40 & 0.80 & 1.05 \\ \hline Total bits/$10^5$ &672 &5.30 &4.84 &21.0& 192& 47.2
    \end{tabular}
    
    \begin{center}
    (b) Watts–Strogatz
    \end{center}
    \begin{tabular}{c||c|c|c|c|c|c}
    Quantizers & No $\Qb$   & $\hat{\Qb}_{\mathrm{sc}}$ & $\hat{\Qb}_{\mathrm{gr}}$& $\Qb_{\mathrm{si}}$&$\Qb_{\rm int}$& $\Qb_{\rm sto}$  \\ \hline
   Iterations/$10^4$ &1.25 &2.67 &1.67 &1.40 & 1.30 & 1.35 \\ \hline Total bits/$10^5$ &1200 &5.34 &6.68 &21.0 & 312& 60.8
    \end{tabular}
    
    \begin{center}
    (c) Barab\'{a}si-Albert
    \end{center}
    \begin{tabular}{c||c|c|c|c|c|c}
    Quantizers & No $\Qb$   & $\hat{\Qb}_{\mathrm{sc}}$ & $\hat{\Qb}_{\mathrm{gr}}$& $\Qb_{\mathrm{si}}$&$\Qb_{\rm sto}$& $\Qb_{\rm int}$  \\ \hline
    Iterations/$10^4$ &0.81 &2.66 &1.22 &1.40 & 1.00&1.05 \\ \hline Total bits/$10^5$ &778 &5.32 &4.88 &21.0& 240 &47.3
    \end{tabular}
    
    \label{tab-net-quantizer}

\end{table}

\begin{figure}[http]
    \centering
    \includegraphics[width=7cm]{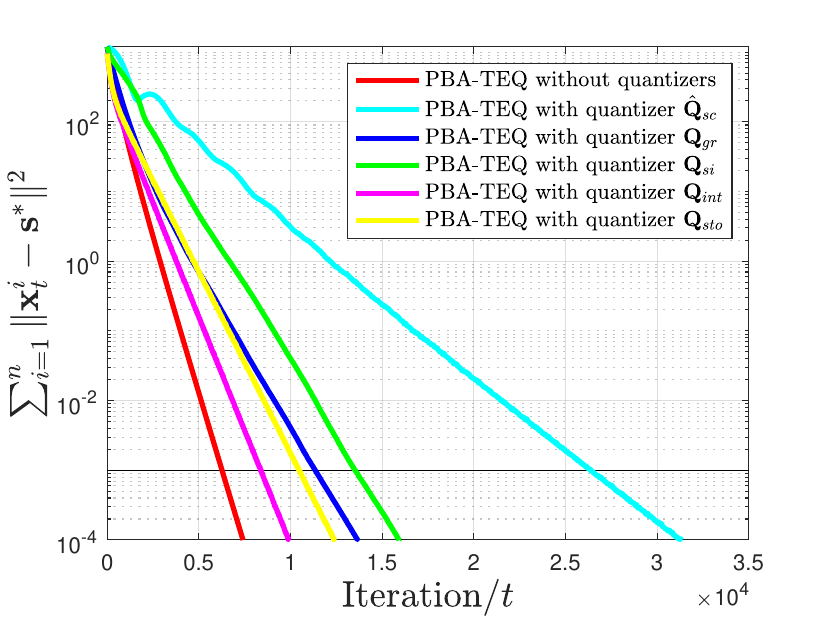}
    \begin{center}
    (a) Erd\H{o}s–R\'{e}nyi
    \end{center}
    \includegraphics[width=7cm]{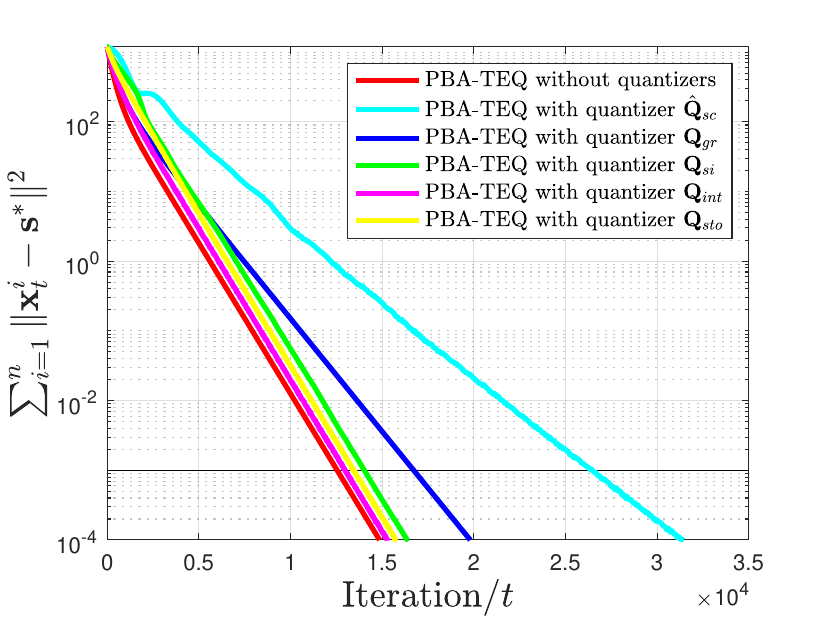}
    \begin{center}
    (b) Watts–Strogatz
    \end{center}
    \includegraphics[width=7cm]{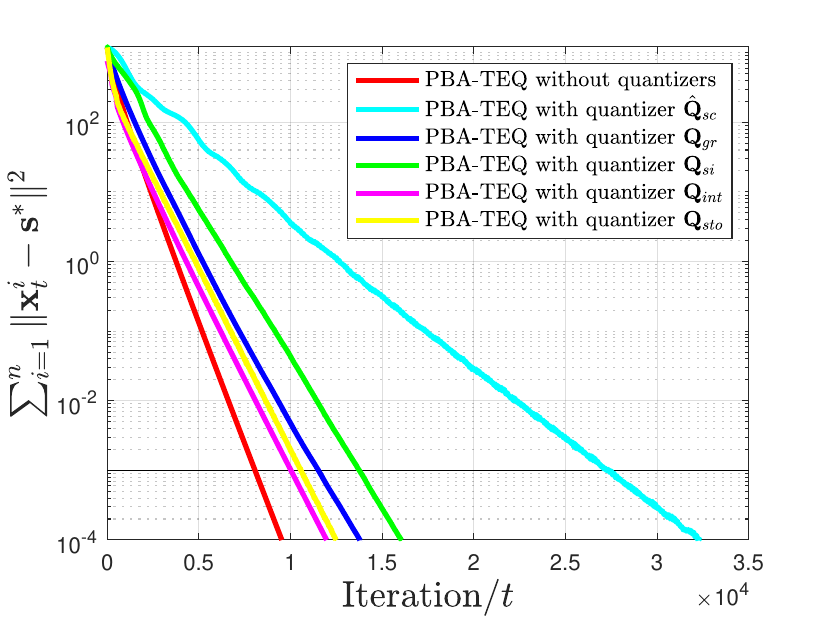}
    \begin{center}
    (c) Barab\'{a}si-Albert
    \end{center}
    \caption{The values $\sum_{i=1}^n\|\xb^i(t)-\sbo\|^2$ along with iterations in Algorithm \ref{alg} without quantizers and with quantizers $\hat{\Qb}_{\mathrm{sc}}$, $\hat{\Qb}_{\mathrm{gr}}$, $\Qb_{\mathrm{si}}$ over different network structures.}
    \label{fig-net-quantizer}
\end{figure}

\textbf{Comparative Analysis of Different Quantization Bits for Quantizer $\hat{\Qb}_{\mathrm{sc}}$}. We further investigate the impact of quantization precision on the convergence performance of Algorithm \ref{alg} by focusing on quantizer \(\hat{\Qb}_{\mathrm{sc}}\). In this experiment, while keeping all other parameters unchanged, we vary the number of quantization bits $m$ and adjust its corresponding constant $\gamma$ accordingly. In contrast, using fewer bits may cause the algorithm to fail to converge, while using more bits, although it can accelerate the convergence rate, will also increase the communication overhead. This empirical observation is in full agreement with the theoretical conclusion presented in Theorem \ref{thm-C1C2}.

\begin{table}[http]
    \centering
    \caption{The relevant data of Algorithm \ref{alg} with $\hat{\Qb}_{\mathrm{sc}}$ for different number of bits for each value over different network structures.}
       
    \begin{center}
    (a) Erd\H{o}s–R\'{e}nyi
    \end{center}
    \begin{tabular}{c||c|c|c|c}
        Bits for each value&64&32&16&8\\ \hline 
        Number of iterations /$10^{4}$&2.65	&2.96	&3.12	&3.17\\ \hline 
        Total bits /$10^{5}$&169.6	&94.7	&49.9	&25.4
   \end{tabular}  
    
    \begin{center}
    (b) Watts–Strogatz
    \end{center}
    \begin{tabular}{c||c|c|c|c}
        Bits for each value&64&32&16&8\\ \hline 
        Number of iterations /$10^{4}$ &2.67	&3.09 &3.18 &3.45\\ \hline 
        Total bits /$10^{4}$ &170.9	&98.9	&50.9	&27.6
   \end{tabular}  
    
    \begin{center}
    (c) Barab\'{a}si-Albert
    \end{center}
    \begin{tabular}{c||c|c|c|c}
        Bits for each value&64&32&16&8\\ \hline 
        Number of iterations /$10^{4}$ &2.66 &2.79 &2.85 &3.03\\ \hline 
        Total bits /$10^{4}$ &170.2	&89.3 &45.6 &24.2
   \end{tabular}  
    
    \label{tab-net-bit}

\end{table}

\section{Conclusion}
In this paper, we have proposed  quantizers in distributed NE seeking algorithms by a general class of ultimate-boundedness-based quantizers, and incorporated them to the PBA. A sufficient condition for maintaining linear convergence of the algorithm has also been established. Then we derived the lower bound of $\log_2(\mathcal{O}(nd))$ for the required bit rate in digital channels when the scalarization quantizer and the greedy quantizer are incorporated, and showed their advantage over the unquantized sparsifiers and sign quantizers.
In the future, we will further explore other ways to combine sparsification and quantization, aiming to investigate the potential for further reducing the required communication bits. Meanwhile, for specific distributed algorithms, we may be able to propose targeted methods to reduce communication. The quantitative relationship between the reduction of communication bits and the convergence rate of the algorithm will also be a focus of our future research.

\appendix

\section*{Appendices}

\section{Proof for Lemma \ref{thm:com}}
\label{app:pro}
\emph{Proof of $\Qb_{\rm sc}$}
: 
First, we prove $\Qb_{\mathrm{sc}}$ is ultimate-boundedness-based.
Letting $\Phi_t^{t+kD}$ be the state transition matrix of the system \begin{equation}
\label{eq:q1}
\xb_{t+1}=\xb_{t}-\kappa_0\psib(t)\psib(t)^{\top}\xb_{t}.\end{equation} 
It is easy to verify that 
$
  \|\Phi_t^{t+j}\xb\|\leq \|\xb\|,
$
for any $j\in[0,kD]$ and $\xb\in\mathbb{R}^{D}$ by noting that $-\psib(t)\psib(t)^{\top}$ is negative definite.
By recalling \cite{LW-DSFN}, \eqref{eq:q1} is uniformly globally linearly stable at zero equilibrium for $\kappa_0=\frac{1}{\alpha_1}$, then the Lyapunov function $V_e(\xb_,t):=\frac{1}{kD}\sum_{j=0}^{kD-1}\|\Phi_t^{t+j}\xb\|^2$ satisfies 
\begin{equation}
    \label{eq:Ve.a}
    \ba    \quad c_1\|\xb\|^2\leq    V_{t}\leq   c_2\|\xb\|^2,&\\
        V_e(\xb-\kappa_0\psib(t)\psib(t)^{\top}\xb,t+1)-V&(\xb,t)
    \leq -2c_3\|\xb_{t}\|^2,
    \ea
\end{equation}
for $c_1=\frac{1}{kD},c_2=1$ and $c_3=\frac{1}{2kD}$.

In addition, we have 
\begin{equation}
\label{eq:c4}
    \ba
&V_e(\xb+\yb,t)-V_e(\xb,t)\\=&\frac{1}{kD}\sum_{j=0}^{kD-1}((\Phi_t^{t+j}(2\xb+\yb))^\top\Phi_t^{t+j}\yb)\\\leq& 2\|\xb\|\|\yb\|+\|\yb\|^2\\
\leq& {c_3}\|\xb\|^2+\left(\frac{1}{c_3}+1\right)\|\yb\|^2.
\ea
\end{equation}
With \eqref{eq:Ve.a} and \eqref{eq:c4} in mind, we assume
\begin{equation}
    \label{eq:q1u}
    \|\xb\|\leq a,
\end{equation}
then
\begin{equation}
     \label{eq:q1d}
\ba
\ &V_e\left(\xb-\kappa_0a\psib(t)q_m\left(\frac{\psib(t)^{\top}\xb}{a}\right),t+1\right)-V_e(\xb,t)\\
\leq &V_e\left(\xb-\kappa_0a\psib(t)q_m\left(\frac{\psib(t)^{\top}\xb}{a}\right),t+1\right)\\&-V_e\left(\xb-\kappa_0a\psib(t)\frac{\psib(t)^{\top}\xb}{a},t+1\right)\\
 &+V_e(\xb-\kappa_0\psib(t)\psib(t)^{\top}\xb,t+1)-V_e(\xb,t)\\
    \leq& -\frac{1}{2kD}\|\xb\|^2+\left({2kD}+1\right)\left(\frac{a}{2^m\alpha_1}\right)^2,
\ea
\end{equation}
where the last inequality uses the fact
\[
\ba
\ &\left\|\kappa_0a\psib(t)q_m\left(\frac{\psib(t)^{\top}\xb}{a}\right)-\kappa_0a\psib(t)\left(\frac{\psib(t)^{\top}\xb}{a}\right)\right\|\\\leq&\kappa_0a\|\psib(t)\|\frac{l}{2}=\frac{\kappa_0al}{2}=\frac{\kappa_0a}{2^m}=\frac{a}{2^m\alpha_1}
\ea
\]
derived from $\|\psib(t)\|=1$ and \eqref{eq:qm} by noting that $|\psib(t)^\top\xb|\leq\|\psib(t)\|\|\xb\|\leq1$ with \eqref{eq:q1u}.


With \eqref{eq:Ve.a}, \eqref{eq:c4}, \eqref{eq:q1u} and \eqref{eq:q1d}, we directly obtain that $\Qb_{\mathrm{sc}}$ is $\left(1,\frac{\sqrt{{2kD}+1}}{2^m\alpha_1}\right)$-ultimate-boundedness-based with $c_1=\frac{1}{kD},c_2=1,c_3=\frac{1}{2kD},c_4=1$.

\emph{Proof of $\Qb_{\rm gr}$}
: 
Next, we prove $\Qb_{\mathrm{gr}}$ is ultimate-boundedness-based. We assume \eqref{eq:q1u} holds and let $V_e(\xb,t)=\|\xb\|^2$, then we have
\begin{equation}
     \label{eq:q2d}
\ba
&V_e\left(\xb-a\Qb_{\rm gr}\left(\frac{\xb}{a}\right)\right)-V_e(\xb)\\
= &\left(\left\|\xb-[\xb_e]_{s}{\bf e}_{s}\right\|^2-\left\|\xb\right\|^2\right)+\left\|a\Qb_{\rm gr}\left(\frac{\xb}{a}\right)\right\|^2 -\left\|[\xb_e]_{s}{\bf e}_{s}\right\|^2\\
\leq & -\frac{1}{2D}\left\|\xb\right\|^2+ (2D+1)\left(\frac{a}{2^m}\right)^2,
\ea
\end{equation}
where the last inequality is obtained by \eqref{eq:qm} by noting that \eqref{eq:q1u} holds.

In addition, we have
\begin{equation}
    \label{eq:c42}
 \ba
&V_e(\xb+\yb,t)-V_e(\xb,t)\\=&(2\xb+\yb))^\top\yb\\\leq& 2\|\xb\|\|\yb\|+\|\yb\|^2.
\ea    
\end{equation}

With \eqref{eq:q2d}, it can be concluded that $\Qb_{\rm gr}$ is $(1,{\frac{\sqrt{2D+1}}{2^m}})$ ultimate-boundedness-based with $c_1=c_2=1,c_3=\frac{1}{D},c_4=1$.

\emph{Proof of $\Qb_{\rm si}$}
: 
Finally, for the sign quantizer $\Qb_{\mathrm{si}}$, noting that 
\begin{equation}
    \label{eq:si}
\|\Qb_{\mathrm{si}}(\xb)-\xb\|_\infty\leq\frac{1}{2}
\end{equation}
if $\|\xb\|_\infty\leq1$. We let $u=\sqrt{n}$. For $\|\xb\|\leq au$ satisfying $\|\frac{x}{a}\|_\infty\leq1$, 
letting $V_e(\xb,t)=\|\xb\|^2$, then we have
\begin{equation}
     \label{eq:q3d}
\ba
\ &V_e\left(\xb-a\Qb_{\rm si}\left(\frac{\xb}{a}\right)\right)-V_e(\xb)\\
= &-2\xb^\top a\Qb_{\rm si}\left(\frac{\xb}{a}\right)+ a^2\left\|\Qb_{\rm si}\left(\frac{\xb}{a}\right)\right\|^2\\
\leq & -2\|\xb\|^2+ \|\xb\|{\sqrt{n}} a^2+a^2\left\|\Qb_{\rm si}\left(\frac{\xb}{a}\right)\right\|^2\\
    \leq& -\|\xb\|^2+\frac{5n}{4}a^2,
\ea
\end{equation}
where the first inequality is obtained by \eqref{eq:si}.
With \eqref{eq:c42}, it can be concluded that $\Qb_{\rm si}$ is $(\sqrt{n},\sqrt{\frac{5n}{4}})$ ultimate-boundedness-based with $c_1=c_2=c_3=c_4=1$.

\section{Proof of Theorem \ref{thm-TSC}}
\label{app:thm2}
By recalling \cite{XY-CCFD}, there holds
$\sigmab^i(t)-\zb^i(t)=\sum_{j\in\mathbb{N}_i}\Lb_{ij}\sigmab^j(t)$, then we know
\eqref{eq:PBA-TEQ} is equal to the following compact form as 
   \begin{equation}
   \label{eq:TSC_compact-d}
        \ba\sigmab_{t+1}&=\sigmab_{t}+\kappa_0g_t\Qc\left(\frac{\yb_t-\sigmab_{t}}{g_t},t\right),\\
\yb_{t+1}&=\yb_t- \alpha \Lb \sigmab_{t}-\eta\Tc \Hc(\yb_t),\\
g_t&=g_0{\gamma^t},\\
        \ea
    \end{equation}
where $\Qc(\frac{\yb_t-\sigmab_t}{g_t}):=[\Qb(\frac{\xb^1(t)-\sigmab^{1}(t)}{g_t},t);\dots;\Qb(\frac{\xb^n(t)-\sigmab^{n}(t)}{g_t},t)]\in \mathbb{R}^{nnd}$, $\sigmab_t:=\sigmab(t)=[\sigmab^{1}(t);\sigmab^{2}(t);\dots;\sigmab^{n}(t)]\in \mathbb{R}^{nnd}$, $\yb_t:=\yb(t)=[\yb^{1}(t);\yb^{2}(t);\dots;\yb^{n}(t)]\in \mathbb{R}^{nnd}$, 
$\Lc:=\Lb\otimes\mathbf{I}_{nd}
$, $\Tc:=\mathrm{blkdiag}(\Tb_1,\dots,\Tb_n)$, 
and $\yb_t$ and $\Hc$ 
denote the estimate vector of all agents and a function related to object functions, respectively, defined in Section \ref{sec.PBNSF}.


By \cite{RACM}, as Assumption \ref{ass-H1} holds, there exists a unique NE $\sbo\in\mathbb{R}^{nd}$ for the game satisfying $\Hc(\mathbf{1}_n\otimes\sbo)=\Hb(\sbo)=\mathbf{0}_{nd}$. Then it can be noticed that $\sigmab^\ast=\yb^\ast=\mathbf{1}_n\otimes\sbo$ is the equilibrium point of system \eqref{eq:TSC_compact-d}.
We introduce the state error by defining $\sil_t:=\sigmab_t-\mathbf{1}_n\otimes\sbo$, $\yl_t:=\yb_t-\mathbf{1}_n\otimes\sbo$. 
Next, we take the state error for \eqref{eq:TSC_compact-d} and yield
\begin{equation}
\ba\label{eq:TSC_0}
{\sil}_{t+1}&=\sil_t+\kappa_0\Qc(\yl_t-\sil_t,t),\\
{\yl}_{t+1}&=\yl_t-\big(\alpha\mathcal{L}\sil_t+\eta \Tc \Hl(\yl_t)\big), \\
g_t&=g_0{\gamma^t},
\ea
\end{equation}
where $\Hl(\yl):=\Hc(\yl+\mathbf{1}_n\otimes\sbo)-\Hc(\mathbf{1}_n\otimes\sbo)$.

To prove the convergence of the system, we will introduce a projection matrix and decompose $\yl_t$.
We let $\mathbf{S}\in\mathbb{R}^{n\times(n-1)}$ be a matrix whose rows are eigenvectors corresponding to nonzero eigenvalues of 
$\Lb$.
Then we decompose $\yl$ by defining $\ych:=\Sc^{\top}\yl_t$ and $\ypi:=\onb^{\top}\yl_t$,
where $\Sc:=\mathbf{S}\otimes \mathbf{I}_{nd}$ and $\onb:=\frac{1}{\sqrt{n}}\mathbf{1}_{n}\otimes \mathbf{I}_{nd}$.
By the fact
$
\Sc^{\top}\onb =\mathbf{0},\quad  \mathbf{I}_{nnd} = \Sc\Sc^{\top} +\onb\onb^{\top},
$
we can derive
\begin{equation}
    \label{eq:dec}
    \ba
    \yl_t = \Sc\ych + \onb\ypi.
    \ea
\end{equation}
Then it can be noticed that the convergence of $\yl_t$ can be shown if $\ypi$ and $\ych$ converge to the zero equilibrium, respectively.

With \eqref{eq:TSC_0} and the fact 
$
    \Lc \onb = \onb^{\top}\Lc = \mathbf{0},
$
we have
\begin{equation}
\ba\label{eq:TSC_D-d}
\sil_{t+1}&=\sil_t+\kappa_0\Qc(\yl_t-\sil_t,t),\\
\ychd&=\ych-\alpha\Sc^{\top}\mathcal{L}\sil_t-\eta\Sc^{\top} \Tc \Hl(\yl_t), \\
\ypid&=\ypi-\eta\onb^{\top} \Tc \Hl(\yl_t), \\
g_t&=g_0{\gamma^t}.
\ea
\end{equation}
Besides, we can also obtain that
\begin{equation}
\ba\label{eq:T_fact}
\|\Lc\sil_t\|^2&\leq 2\lambda^2_n\|\yl_t-\sil_t\|^2 + 2\lambda^2_n\|\ych\|^2.
\ea
\end{equation}

Now we are ready to propose Lyapunov functions for system \eqref{eq:TSC_D-d}. 
Define $V_{1,t}:=\frac{1}{2}\|\ych\|^2$, then 
\begin{equation}
\label{eq:V1}
\ba
&\quad {V}_{1,t+1}-V_{1,t} \\
&\leq  \big(-\frac{1}{2} \alpha \lambda_2 \|\ych\|^2 +  \frac{1}{2} \alpha \lambda_n \|\sil_t-\yl\|^2
\\
&\quad +\frac{1}{2}\eta(1+L^2_\Hc)\|\ych\|^2+\frac{1}{2}\eta L^2_\Hc \|\ypi\|^2\big)\\&\quad 
+\big(\alpha^2\|\Lc\sil_t\|^2+\eta^2L^2_\Hc 
(\|\ych\|^2+\|\ypi\|^2)]\big)\\
&\leq  \big( -\frac{1}{2} \alpha \lambda_2 \|\ych\|^2 +  \frac{1}{2} \alpha \lambda_n  \|\yl_t-\sil_t\|^2
\\
&\quad +\frac{1}{2}\eta(1+L^2_\Hc)\|\ych\|^2+\frac{1}{2}\eta L^2_\Hc \|\ypi\|^2\big)\\&\quad +[2\alpha^2\lambda_n^2 \|\yl_t-\sil_t\|^2+2\alpha^2\lambda_n^2\|\ych\|^2\\&\quad +\eta^2L^2_\Hc (\|\ych\|^2+\|\ypi\|^2)],
\ea
\end{equation}
where the first inequality is obtained by 
\[
    \ba
    \Lc \yl_t = \Lc (\Sc \Sc^{\top}+\onb\onb^{\top} )\yl_t = \Lc \Sc \ych,
    \ea
\]
\begin{equation}
    \label{eq:HY}
    \ba
    \|\Hl(\yl_t)\|^2\leq L^2_\Hc \|\yl_t\|^2=L^2_\Hc (\|\ych\|^2+\|\ypi\|^2),
    \ea
\end{equation}
derived from \eqref{eq:ass-H2} in Assumption \ref{ass-H2} and $\|\Tc\|\leq 1$, and the last inequality is obtained by \eqref{eq:T_fact}.

Define $V_{2,t}:=\frac{1}{2}\|\ypi\|^2$, then 
\begin{equation}
    \label{eq:V2}
    \ba
    &\quad V_{2,t+1}-V_{2,t}
    \\&\leq  \big(-\eta (\onb\ypi)^{\top}\Tc[\Hc(\yl_t+\mathbf{1}_n\otimes\sbo)\\&\quad -\Hc(\onb\ypi+\mathbf{1}_n\otimes\sbo)+\Hc(\onb\ypi+\mathbf{1}_n\otimes\sbo)\\
    &\quad -\Hc(\mathbf{1}_n\otimes\sbo)]\big)+\eta^2\|\Tc\Hl(\yl_t)\|^2\\
    &\leq  \big(-\eta\frac{\mu}{2n} \|\ypi\|^2+\eta\frac{n}{2\mu}L_\Hc^2 \|\ych\|^2\big)\\&\quad +\eta^2L^2_\Hc (\|\ych\|^2+\|\ypi\|^2),
    \ea
\end{equation}
where the second inequality is obtained by \eqref{eq:dec}, \eqref{eq:HY}, $\|\Tc\|\leq1$, $
    \onb^{\top}\Tc= \frac{1}{\sqrt{n}}\mathbf{I}_{nd},
$
the fact $
\Hc(\mathbf{1}_n\otimes\xb)
=\Hb(\xb)
$ for any $\xb\in\mathbb{R}^{nd}$ and \eqref{eq:ass-H1} in Assumption \ref{ass-H1}.

As $\Qb$ is ultimate-boundedness-based for some $u$, let's assume that
\begin{equation}
    \label{eq:condi}
\left\|\frac{\yb_t-\sigmab_t}{g_t}\right\|\leq u
\end{equation}
holds for this system (it will be proven later).
Then, by Definition \ref{def:com}, there exists a Lyapunov function $V_e$ satisfying \eqref{eq:def2}. Define $V_3(\yl_t-\sil_t,t)=V_3(\yb_t-\sigmab_t,t):=\sum_{i=1}^n V_e(\xb^i(t)-\sigmab^i(t),t)$, with \eqref{eq:PBA-TEQ},  then we have
\begin{equation}
    \label{eq:V3-d}
    \ba {}
    &\ \quad V_{3,t+1}-V_{3,t} 
    \\&=\sum_{i\in\mathbb{V}}(V_{e}(\xb^i(t+1)-\sigmab^i(t+1),t+1)\\&\quad\ -V_{e}(\xb^i(t)-\sigmab^i(t),t))
    \\
    &\leq  -c_3 \|\yl_t-\sil_t\|^2 + 2c_4\|\yl_t-\sil_t\|\| \alpha\mathcal{L}\sil_t+\eta \Tc \Hl(\yl_t)\| \\&\quad +c_4\|\alpha\Lc\sil_t+\eta\Tc\Hl(\yl_t)\|^2+\delta g_t^2\\
    &\leq -[c_3-\big( -
c^2_4 \alpha/r-c_4^2 \eta/r\\&\quad -2 \alpha\lambda_n^2 r\big)] \|\yl_t-\sil_t\|^2 +\big(( 2 \alpha\lambda_n^2 r\\
    &\quad + \eta rL_\Hc^2)\|\ych\|^2+ \eta rL_\Hc^2\|\ypi\|^2\big)
    \\&\quad +c_4(4\alpha^2\lambda_n^2(\|\yl_t-\sil_t\|^2+\|\ych\|^2)\\&\quad +2\eta^2L^2_\Hc (\|\ych\|^2+\|\ypi\|^2)\big)
    +\delta^2g_t^2,
    \ea
\end{equation}
where the first inequality is obtained by \eqref{eq:def2}, and the last inequality is obtained by \eqref{eq:T_fact}, \eqref{eq:HY} and Young's Inequality, with $r>0$ being an undetermined parameter to be chosen later.

Now we introduce some parameters $\xi_1, \xi_2,\dots,\zeta_1,\zeta_2,\dots>0$ independent of $\alpha$ and $\eta$, as follows
\[
\ba
&\xi_1=\frac{\lambda_2}{2},\quad\xi_2=\frac{3+3L_\Hc^2}{2}+\frac{6n^2L_\Hc^4}{\mu^2},\\
&
\xi_3=2 
\lambda_n^2,\\
&\xi_4=2c_4^2 ,\quad\xi_5=\frac{7\lambda_n}{12},\\
& \zeta_1=2\lambda_n^2+(1+2c_4+p)L_\Hc^2+4c_4\lambda_n^2,
\\
&\zeta_2=(1+2c_4+p))L_\Hc^2,\quad\zeta_3=2\lambda_n^2+4c_4\lambda_n^2,
\ea
\]
where $p=\frac{12L^2_\Hc n}{\mu}>0$.
Then we define the total Lyapunov functions of system \eqref{eq:TSC_0} as $ V_t: = V_{1,t}+pV_{2,t}+V_{3,t}$. By \eqref{eq:Ve.a}, there holds
\begin{equation}
\ba\label{eq:V_posi}
V_t\leq \frac{1}{2}\|\ych\|^2+\frac{p}{2}\|\ypi\|^2+  c_2\|\yl_t-\sil_t\|^2.
\ea
\end{equation}

We let
$
r\leq1, \quad \eta\leq\alpha
$
to simplify the following process,
then by \eqref{eq:V1}, \eqref{eq:V2} and \eqref{eq:V3-d}, we have
\[
    \ba
    &\quad  V_{t+1}-V_t\\
    &\leq  -(\xi_{1} \alpha -\xi_{2} \eta - \xi_3 \alpha r)\|\ych\|^2
    \\
    &\quad -(p\eta\frac{\mu}{4n}) \|\ypi\|^2-({c_3}-\xi_4 \alpha/r - \xi_5 \alpha)\|\yl_t-\sil_t\|^2\\&\quad +(\alpha^2\zeta_1\|\ych\|^2+\eta^2\zeta_2\|\ypi\|^2+\alpha^2\zeta_3\|\yl_t-\sil_t\|^2)+ \delta^2g^2_t.
    \ea
\]

In succession, we choose $r=\mathrm{min}\left\{\frac{\xi_1}{3\xi_3},1\right\}$ independent of $\alpha,\eta$, $\alpha=c_3\mathrm{min}\left\{\frac{r}{4\xi_4},\frac{1}{4\xi_5}\right\}$ independent of $\eta$, $\eta=\alpha\mathrm{min}\left\{\frac{\xi_1}{3\xi_2},1\right\}$, then we have
\[
\ba
    &\quad  V_{t+1}-V_t\\&\leq   -(\frac{\xi_1\alpha}{3}-\alpha^2\zeta_1)\|\ych\|^2 -(p\eta\frac{\mu}{4n}-\eta^2\zeta_2) \|\ypi\|^2\\&\quad -(\frac{c_3}{2}-\alpha^2\zeta_3) \|\yl_t-\sil_t\|^2+ \delta^2g^2_t,
    \ea
    \]

Letting $\alpha\leq \min\{\frac{\xi_1}{6\zeta_1},\}$
    \[
    \ba
    &\leq  -\mathrm{min}\left\{\frac{\xi_1\alpha}{3},\frac{\eta\mu}{4n},\frac{c_3}{4c_2 }\right\} V_{t}+ \delta^2g^2_t\\
    &= -\beta V_t+  \delta^2g^2_t,
    \ea
\]
where $\beta:=\mathrm{min}\left\{\frac{\lambda_2\alpha}{6},\frac{\eta\mu}{4n},\frac{c_3}{4c_2 }\right\}\in(0,1)$ and the last inequality is obtained by \eqref{eq:V_posi}. 
Then we have
\begin{equation}
\label{T_V_nega}
\ba
V_t&\leq V_0(1-\beta)^t+ g^2_0\delta^2\sum_{\tau=0}^t\frac{(1-\beta)^{t-\tau}}{\gamma ^{2\tau}}\\&\leq V_0 (1-\beta)^t+\frac{ g^2_0\delta^2}{1-\gamma^2/{(1-\beta)}}
((1-\beta)^ t-\gamma^{2t}).
 \ea
\end{equation}

The result \eqref{T_V_nega} holds as we assume \eqref{eq:condi} in the process of the system's dynamic. Next, we will prove this fact.
By \eqref{eq:Ve.a}, we derive
\[
\ba
\left\|\frac{\yl_t-\sil_t}{g_t}\right\|^2 &\leq \frac{{V_t}}{ g^2_t}\\&\leq\frac{{V_0}}{ g^2_0}(\frac{1-\beta}{\gamma^2})^t+\frac{ \delta^2}{ 1-{\gamma^{2}/(1-\beta)}}((\frac{1-\beta}{\gamma^2})^t-1).
\ea
\]

We let $\gamma^2=\frac{1-\beta}{1-\frac{\beta}{1+\epsilon}}$ and $g^2_0=\frac{V_0(\gamma^2/(1-\beta)-1)}{ \delta^2}$, then we have 
\[
\left\|\frac{\yl_t-\sil_t}{g_t}\right\|^2 \leq \frac{ \delta^2}{ 1/(1-\frac{1}{1+\epsilon}\beta)-1}\leq \frac{1+\epsilon}{\beta}\delta^2.
\]

It can be furthermore noticed that the \eqref{eq:condi} will hold if $u$ and $\delta$ satisfy the following condition 
\begin{equation}
    \label{eq:uanddelta_O}
    \frac{u}{\delta}\geq\sqrt{\frac{1+\epsilon}{\beta}}.
\end{equation}

To sum up, if $\Qb$ satisfies  \eqref{eq:uanddelta_O}, then \eqref{eq:condi} holds. 
Then \eqref{T_V_nega} holds and $V_t$ converges to zero linearly. 
With \eqref{eq:V_posi},
we can derive 
$\|\ych\|$ and $\|\ypi\|$ linearly converge to the zero equilibrium and so does $\|\yl_t\|$ by \eqref{eq:dec}. 
With the definition of 
$\yb_t$ and
$\yl_t=\yb_t-\mathbf{1}_n\otimes \sbo$  before, we know $\xb^i(t)$ in Algorithm \ref{alg} reaches the NE $\sbo$ linearly.
Then Theorem \ref{thm-TSC} is proven.

\section{Proof of Theorem \ref{lem-Q3}}
\label{appc}

According to proof in Appendix \ref{app:pro}, \eqref{eq:def2} holds if the following holds.
\begin{equation}
    \label{eq:condi2}
    \left\|\frac{\yl_t-\sil_t}{g_t}\right\|_\infty\leq 1.
\end{equation}
Then, similar to the proof of Theorem \ref{thm-TSC} in Section \ref{sec.conana}, we can obtain \eqref{T_V_nega}.
Therefore, the proof can be completed if we can prove \eqref{eq:condi2} for all $t\in\mathbb{N}_+$ when $\Qb_{\mathrm{si}}$ is applied. By $\yb_0=\sigmab_0=\mathbf{0}$, $\yl_t-\sil_t=\yb_t-\sigmab_t$, $g_0>0$ and $u>0$, we have $\left\|\frac{\yl_0-\sil_0}{g_0}\right\|_\infty\leq 1$. Next, we prove if \eqref{eq:condi2} holds for $t\geq0$, there holds $\left\|\frac{\yl_{t+1}-\sil_{t+1}}{g_{t+1}}\right\|_\infty\leq 1$. 

By \eqref{eq:TSC_0}, we have
\[
\ba
&\quad \left\|\frac{\yl_{t+1}-\sil_{t+1}}{g_{t+1}}\right\|_\infty\\
&=\left\|\frac{\dfrac{\yl_{t}-\sil_{t}}{g_t}-\Qb_{\mathrm{si}}\left(\frac{\yl_t-\sil_t}{g_t}\right)}{g_{t+1}/g_t}\right\|_\infty+\left\|\frac{\alpha \Lb \sil_{t}+\eta\Tc \Hl(\yl_t)}{g_{t+1}}\right\|\\
&\leq \frac{2\alpha}{g_{t+1}} 
{(}\lambda^2_n
\left\|\yl_t-
\sil_t\right\|^2+ \lambda^2_n\left\|\ych\right\|^2\\&\quad\ +L^2_\Hc \|\ych\|^2+L^2_\Hc\|\ypi\|^2)^{\frac{1}{2}}+\frac{1}{2\gamma}\\
&\leq\frac{2\alpha c_5}{g_0 \gamma}V_0^{\frac{1}{2}}+\frac{1}{2\gamma},
\ea\]
with $c_5=\sqrt{2\mathrm{max}\left\{(\lambda^2_n+L^2_\Hc),\frac{L^2_\Hc}{p}\right\}}$, where the first inequality is obtained by \eqref{eq:si}, \eqref{eq:T_fact} and \eqref{eq:HY}, and the last inequality is obtained by \eqref{eq:V_posi}, \eqref{T_V_nega} and letting $g^2_0={4V_0(\gamma^2/(1-\beta)-1)}$.
Then we choose $\gamma\in (\max\{\frac{1}{2},\sqrt{1-\beta}\},1)$ and $\alpha<\frac{(\gamma-0.5)g_0}{2 c_5\sqrt{V_0}}$. It can be obtained that $\left\|\frac{\yl_{t+1}-\sil_{t+1}}{g_{t+1}}\right\|_\infty\leq 1$ holds if $\left\|\frac{\yl_t-\sil_t}{g_t}\right\|_\infty\leq 1$. 
With the fact $\left\|\frac{\yl_0-\sil_0}{g_0}\right\|_\infty\leq 1$, we can conclude \eqref{eq:condi2} holds for any $t\geq0$. 
Then \eqref{T_V_nega} holds with $\gamma<1$ and we complete the proof. 

 \section{Proof of Theorem \ref{thm-C1C2}}
 \label{pr:Cor1}

 With the values of the parameters $(u,\delta)$ for different quantizers in Appendix \ref{app:pro} in mind, we analyze the condition \eqref{eq:uanddelta_O_INI} for $\Qb_{\rm sc}$ and $\Qb_{\rm gr}$, respectively. 
For $\Qb_{\rm sc}$, there holds $\frac{u}{\delta}=\frac{2^m\alpha_1}{\sqrt{{2kD}+1}}$, the condition \eqref{eq:uanddelta_O_INI} is 
\[
m\geq \mathrm{log}_2 \left(\frac{1}{\alpha_1}\sqrt{\frac{(1+\epsilon)(2knd+1)}{\beta} }\right).
\]
Further, by the expression of $\beta$ in Appendix, it can be easily noticed that $\beta\leq \frac{c_3}{4c_2}$. With $c_2=1,c_3=\frac{1}{2knd}$ for $\Qb_{\rm sc}$ in Appendix \ref{app:pro}, the condition \eqref{eq:uanddelta_O_INI} is satisfied if
\[
\ba
m\geq \underline{m}&= \mathrm{log}_2 \left(\frac{1}{\alpha_1}\sqrt{{(1+\epsilon)((2knd)^2+2knd)} }\right)\\&=\log_2(\mathcal{O}(nd)).
\ea
\]

Similarly, for $\Qb_{\rm gr}$, there holds $\frac{u}{\delta}=\frac{2^m}{\sqrt{{2nd}+1}}$, $c_2=1$ and $c_3=\frac{1}{2nd}$. Then the condition \eqref{eq:uanddelta_O_INI} is satisfied if
\[
\ba
m\geq \underline{m}&= \mathrm{log}_2 \left(\sqrt{{(1+\epsilon)((2nd)^2+2nd)} }\right)\\&=\log_2(\mathcal{O}(nd)).
\ea
\]
Now that \eqref{eq:uanddelta_O_INI} is satisfied, according to Theorem \ref{thm-TSC}, we complete the proof.


\end{document}